\documentclass[preprint,12pt]{elsarticle}

\usepackage{amsmath,amssymb,amsthm,mathtools}

\newtheorem{theorem}{Theorem}[section]
\newtheorem{proposition}[theorem]{Proposition}
\newtheorem{lemma}[theorem]{Lemma}

\newtheorem{conjecture}[theorem]{Conjecture}
\theoremstyle{definition}
\newtheorem{definition}[theorem]{Definition}
\newtheorem{example}[theorem]{Example}
\theoremstyle{remark}

\newcommand{\F}{\mathbb{F}}
\newcommand{\PP}{\mathcal{P}}                 % projective space P_q(n)
\newcommand{\G}{\mathcal{G}}                   % Grassmannian
\newcommand{\C}{\mathcal{C}}                   % a code
\newcommand{\dist}{d_S}                        % subspace distance
\newcommand{\boxsum}{\boxplus}                 % linear addition

\journal{Finite Fields and Their Applications}

\begin{document}

\begin{frontmatter}

%%% TODO: confirm title %%%
\title{A proof of the Braun--Etzion--Vardy bound
       for binary subspace codes}

\author[mse]{Srikanth B. Pai\corref{cor1}}
\ead{srikanthbpai@mse.ac.in}
\cortext[cor1]{Corresponding author.}

\address[mse]{Madras School of Economics, Kotturpuram, Chennai 600025, India}

\begin{abstract}
The notion of a linear subspace code in a projective space was introduced by
Braun, Etzion and Vardy (2013), who conjectured that a linear subspace
code in the projective space $\PP_2(n)$ has at most $2^n$ codewords. We resolve this conjecture. A novel character-theoretic argument is used to prove the conjecture and the proof is self-contained.
\end{abstract}

\begin{keyword}
Subspace codes \sep Linear subspace codes \sep Projective space \sep
Braun--Etzion--Vardy conjecture \sep Character theory \sep Gram matrix
\MSC[2020] 94B60 \sep 51E20
\end{keyword}

\end{frontmatter}

%=====================================================================
\section{Introduction}
\label{sec:intro}
%=====================================================================
Let $\F_q$ be the finite field with $q$ elements and let $\PP_q(n)$ be the
projective space of all subspaces of $\F_q^n$. Equipped with the \emph{subspace
distance}
\[
 \dist(X,Y)=\dim X+\dim Y-2\dim(X\cap Y),
\]
$\PP_q(n)$ becomes a metric space, and a \emph{subspace code} is a subset of
$\PP_q(n)$. Such codes were introduced by K\"otter and
Kschischang~\cite{KoetterKschischang} in the context of random network coding,
and their combinatorics was developed by Etzion and Vardy~\cite{EtzionVardy}.

Classical error correcting codes can be studied as an algebraic theory primarily because classical block codes are \emph{linear} and the Hamming distance is translation invariant with respect to this linear structure. Seeking an analogue in projective space, Braun, Etzion and Vardy~\cite{Braun} called a subspace code $\C$ containing $0$ a \emph{linear subspace code} if it carries a binary operation $\boxsum$ making $(\C,\boxsum)$ an elementary abelian $2$-group with identity $0$, under which the subspace distance is translation invariant. Such a code is an
$\F_2$-vector space, so $|\C|$ is a power of $2$. They showed that taking the spans of all subsets of a fixed basis of $\F_2^n$ (with symmetric difference of index sets as addition) yields a linear subspace code of size exactly $2^n$. Braun, Etzion and Vardy conjectured that this construction is optimal: over the binary field, no linear subspace code has more than $2^n$ codewords. We resolve this conjecture in this paper.

In the intervening years, various attempts have been made to partially resolve the conjecture. Pai and Rajan~\cite{PaiRajan} proved the bound under the hypothesis that the full space $\F_2^n$ is itself a codeword, together with the sharp converse that equality then forces $\C$ to be derived from a fixed basis. Their argument uses the fact that translating a codeword by the full space produces a complement of the codeword. This translation map therefore pairs complementary dimension layers, and an application of a cross-intersection
inequality of Lov\'asz finishes the problem. Basu and Kashyap~\cite{BasuKashyapNCC,BasuKashyapLattice} established the bound for codes closed under intersection and analysed their lattice structure. In a different direction, Basu~\cite{Basu} studied equidistant linear subspace codes, which attain the largest normalized minimum distance in his sense. He proved that an equidistant code in $\PP_2(n)$ has at most $2^n$ codewords; for $n\ge2$, equality occurs only for the Fano-plane code in dimension three. He also gave constructions showing that, for every prime power $q>2$, some $\PP_q(n)$ contains an equidistant linear subspace code with more than $2^n$ codewords. Thus the conjecture is false over every non-binary field, and only the binary case can survive.
Most recently, Mahak and Bhaintwal~\cite{MahakBhaintwal} determined the size of a linear subspace code with exactly $n-1$ or $n-2$ one-dimensional codewords.

Each of these results confirms the BEV conjecture under an extra
hypothesis: that the full space is a codeword, that the code is closed under
intersection, that it is equidistant, or that it has many one-dimensional
codewords. The general binary case, reached by none of these methods, has remained open. We settle it in full.

\begin{theorem}
\label{thm:main-intro}
Every binary linear subspace code $\C\subseteq\PP_2(n)$ satisfies
$|\C|\le 2^n$.
\end{theorem}

A high level sketch of the argument is as follows: We embed each codeword $X$
as the normalized indicator $2^{-\dim X/2}\mathbf 1_X$ of its point set in
$\mathbb R^{\F_2^n}$; the bound $|\C|\le 2^n$ follows once these vectors are
shown to be linearly independent, which amounts to proving the nonsingularity of their Gram matrix $M=\bigl(2^{-\dist(X,Y)/2}\bigr)$. Because $M$ is constant on the translates of $\boxsum$, it is diagonalized by the characters of the group
$(\C,\boxsum)$, and a vanishing eigenvalue would force an $\F_2$-character
relation among the indicators. Comparing rational parts of this relation in the
quadratic field $\mathbb Q(\sqrt2)$ then yields a parity contradiction. Hence
we conclude $M$ is nonsingular and the bound is proven.

The paper is organized as follows. Section~\ref{sec:prelim} fixes notation and
recalls the one property of linear subspace codes we use. Section~\ref{sec:proof}
sets up the Gram matrix, works out its spectrum through the character theory of
$(\C,\boxsum)$, and gives the parity argument that proves
Theorem~\ref{thm:main-intro} (restated as Theorem~\ref{thm:main}); it also
discusses the equality case. Section~\ref{sec:remarks} collects concluding
remarks.

%=====================================================================
\section{Preliminaries on linear subspace codes}
\label{sec:prelim}
%=====================================================================
Throughout, $q$ denotes a prime power and $\F_q$ the finite field with $q$
elements. For a positive integer $n$, the vector space $\F_q^n$ has dimension
$n$ over $\F_q$. The \emph{projective space} $\PP_q(n)$ is the set of all
subspaces of $\F_q^n$, and for $0\le k\le n$ the \emph{Grassmannian}
$\G_q(n,k)\subseteq\PP_q(n)$ is the set of subspaces of dimension $k$. We
write $0$ for the trivial subspace $\{0\}$. For subspaces $X,Y\in\PP_q(n)$,
the sum $X+Y$ is the smallest subspace containing both, and the dimension
formula reads
\[
 \dim(X+Y)=\dim X+\dim Y-\dim(X\cap Y).
\]

A \emph{subspace code} is a nonempty subset of $\PP_q(n)$. Following
K\"otter and Kschischang~\cite{KoetterKschischang}, $\PP_q(n)$ is made a
metric space by the \emph{subspace distance}
\[
 \dist(X,Y):=\dim(X+Y)-\dim(X\cap Y)
           =\dim X+\dim Y-2\dim(X\cap Y),
\]
which is used to correct errors and erasures in random network coding.

Braun, Etzion and Vardy~\cite{Braun} singled out those subspace codes that
carry an algebraic structure compatible with the metric.

\begin{definition}
\label{def:linear-code}
A subspace code $\C\subseteq\PP_q(n)$ with $0\in\C$ is a \emph{linear
subspace code} if it is equipped with a binary operation
$\boxsum:\C\times\C\to\C$ such that
\begin{enumerate}
\item[\textup{(i)}] $(\C,\boxsum)$ is an abelian group with identity $0$;
\item[\textup{(ii)}] $X\boxsum X=0$ for every $X\in\C$; and
\item[\textup{(iii)}] the subspace distance is translation invariant, that is,
\[
 \dist(X,Y)=\dist(X\boxsum Z,\,Y\boxsum Z)
 \qquad\text{for all }X,Y,Z\in\C.
\]
\end{enumerate}
The operation $\boxsum$ is called a \emph{linear addition} on $\C$.
\end{definition}

Conditions (i) and (ii) say that $(\C,\boxsum)$ is an elementary abelian
$2$-group; equivalently, $\C$ is a vector space over $\F_2$ under $\boxsum$,
with scalar multiplication $0\cdot X=0$ and $1\cdot X=X$. In particular the
cardinality of a linear subspace code is a power of $2$. Braun, Etzion and
Vardy studied linearity over the binary field and asked whether this power can
ever exceed $2^n$; they conjectured that it cannot
\cite[Section~5, Problem~1]{Braun}.

\begin{conjecture}
\label{conj:bev}
Every binary linear subspace code $\C\subseteq\PP_2(n)$ satisfies
\[
 |\C|\le 2^n.
\]
\end{conjecture}

We refer to Conjecture~\ref{conj:bev} as the
BEV conjecture. Over larger fields the analogous bound
fails~\cite{Basu}: for every prime power $q>2$ there exist linear subspace
codes in $\PP_q(n)$ with more than $2^n$ codewords. 

The only property of the linear addition we shall need is that it realizes the
subspace distance as a dimension:
\begin{equation}
 \dim(X\boxsum Y)=\dist(X,Y)\qquad(X,Y\in\C),
 \label{eq:dim-identity}
\end{equation}
which is \cite[Lemma~6]{Braun}.

From Section~\ref{sec:proof} onwards we specialize to $q=2$ and write
$V=\F_2^n$.

%=====================================================================
\section{Proof of the Braun--Etzion--Vardy conjecture over $\F_2$}
\label{sec:proof}
%=====================================================================

Recall that $V=\F_2^n$. To each subspace $X\in\PP_2(n)$ we associate its
\emph{indicator vector} $\mathbf 1_X\in\mathbb R^{V}$, given by
\[
 \mathbf 1_X(v)=
 \begin{cases}
  1, & v\in X,\\[2pt]
  0, & v\notin X,
 \end{cases}
\]
and we equip $\mathbb R^{V}$ with the standard inner product
\[
 \langle f,g\rangle=\sum_{v\in V} f(v)\,g(v).
\]

A $k$-dimensional subspace of $\F_2^n$ has exactly $2^{k}$ elements. Hence, for
all $X,Y\in\PP_2(n)$,
\begin{equation}
 \langle \mathbf 1_X,\mathbf 1_Y\rangle=|X\cap Y|=2^{\dim(X\cap Y)}.
 \label{eq:indicator-inner}
\end{equation}
In particular $\langle \mathbf 1_X,\mathbf 1_X\rangle=2^{\dim X}$, so we
normalize the indicators and set
\begin{equation}
 u_X:=2^{-\dim X/2}\,\mathbf 1_X\in\mathbb R^{V}
 \qquad(X\in\C).
 \label{eq:normalized}
\end{equation}
Combining \eqref{eq:indicator-inner} with the subspace-distance identity
\[
 \dist(X,Y)=\dim X+\dim Y-2\dim(X\cap Y),
\]
we obtain, for all $X,Y\in\C$,
\begin{equation}
 \begin{aligned}
  \langle u_X,u_Y\rangle
   &=2^{-(\dim X+\dim Y)/2}\,\langle \mathbf 1_X,\mathbf 1_Y\rangle\\[6pt]
   &=2^{-(\dim X+\dim Y)/2}\,2^{\dim(X\cap Y)}\\[6pt]
   &=2^{-\dist(X,Y)/2}.
 \end{aligned}
 \label{eq:gram-entry}
\end{equation}

Recall that the \emph{Gram matrix} of a finite family of vectors in a real
inner-product space is the matrix of their pairwise inner products; it is
positive semidefinite, and it is nonsingular if and only if the vectors are
linearly independent (see, e.g., \cite[Section~7.2]{HornJohnson}). By
\eqref{eq:gram-entry}, the Gram matrix of the family $(u_X)_{X\in\C}$ is
\begin{equation}
 M=\Bigl(\,2^{-\dist(X,Y)/2}\,\Bigr)_{X,Y\in\C},
 \label{eq:gram-matrix}
\end{equation}
and its diagonal entries equal $1$, since $\dist(X,X)=0$.

\begin{proposition}
\label{prop:gram}
If the Gram matrix $M$ of \eqref{eq:gram-matrix} is nonsingular, then
$|\C|\le 2^n$.
\end{proposition}

\begin{proof}
By \eqref{eq:gram-entry}, $M$ is the Gram matrix of the family
$(u_X)_{X\in\C}$ in $\mathbb R^{V}$.

Since the Gram matrix $M$ is nonsingular, the vectors $\{u_X:X\in\C\}$ are
linearly independent in $\mathbb R^{V}$. By \eqref{eq:normalized} each $u_X$ is
a nonzero scalar multiple of $\mathbf 1_X$, so the indicators
$\{\mathbf 1_X:X\in\C\}$ are linearly independent as well. A linearly
independent subset of $\mathbb R^{V}$ has at most
\[
 \dim_{\mathbb R}\mathbb R^{V}=|V|=2^{n}
\]
elements, and therefore $|\C|\le 2^{n}$.
\end{proof}

Being a Gram matrix, $M$ is positive semidefinite, so its nonsingularity is
equivalent to positive definiteness, $M\succ 0$. By
Proposition~\ref{prop:gram} it therefore suffices to prove that $M$ is
nonsingular.

The matrix $M$ carries extra structure that makes this tractable: by
\eqref{eq:gram-entry} and \eqref{eq:dim-identity} its entries depend only on
the group element $X\boxsum Y$,
\begin{equation}
 M_{X,Y}=2^{-\dist(X,Y)/2}=2^{-\dim(X\boxsum Y)/2}.
 \label{eq:group-matrix}
\end{equation}
A matrix of this form is diagonalized by the characters of the abelian group
$(\C,\boxsum)$, which we now recall.

\begin{definition}
A \emph{character} of a finite abelian group $(G,+)$ is a homomorphism
$\chi:G\to\mathbb C^{\times}$. If every element of $G$ has order two, then
\[
 \chi(g)^2=\chi(g+g)=\chi(0)=1,
\]
so $\chi(g)\in\{+1,-1\}$ for every $g\in G$; in particular the characters are
real-valued.
\end{definition}

We use only the following standard facts; see, e.g., \cite{Serre}.

\begin{lemma}
\label{lem:char-orth}
An elementary abelian $2$-group $G$ has exactly $|G|$ characters, and they
satisfy
\[
 \sum_{g\in G}\chi(g)\psi(g)=
 \begin{cases}
  |G|, & \chi=\psi,\\[2pt]
  0, & \chi\ne\psi.
 \end{cases}
\]
Consequently the character vectors $\bigl(\chi(g)\bigr)_{g\in G}$ form an
orthogonal basis of $\mathbb R^{G}$.
\end{lemma}

\begin{proposition}
\label{prop:eigen}
For each character $\chi$ of $(\C,\boxsum)$, the vector
$\bigl(\chi(X)\bigr)_{X\in\C}$ is an eigenvector of $M$ with eigenvalue
\begin{equation}
 \mu_\chi=\sum_{Z\in\C}2^{-\dim Z/2}\,\chi(Z).
 \label{eq:eigval}
\end{equation}
These eigenvectors form a basis of $\mathbb R^{\C}$; hence $M$ is nonsingular
if and only if $\mu_\chi\ne0$ for every character $\chi$.
\end{proposition}

\begin{proof}
Fix $X\in\C$. As $Y$ ranges over $\C$, so does $Z:=X\boxsum Y$, and the
homomorphism property gives $\chi(Y)=\chi(X\boxsum Z)=\chi(X)\,\chi(Z).$

Hence, by \eqref{eq:group-matrix},
\[
 \begin{aligned}
  (M\chi)_X
   &=\sum_{Y\in\C}2^{-\dim(X\boxsum Y)/2}\,\chi(Y)\\[6pt]
   &=\chi(X)\sum_{Z\in\C}2^{-\dim Z/2}\,\chi(Z)\\[6pt]
   &=\mu_\chi\,\chi(X).
 \end{aligned}
\]
Thus $\bigl(\chi(X)\bigr)_{X\in\C}$ is an eigenvector of $M$ with eigenvalue
$\mu_\chi$. By Lemma~\ref{lem:char-orth} these eigenvectors form a basis of
$\mathbb R^{\C}$, so $M$ is diagonalizable with spectrum $\{\mu_\chi\}$ and is
nonsingular precisely when no $\mu_\chi$ vanishes.
\end{proof}

The trivial character $\chi\equiv1$ has
\[
 \mu_\chi=\sum_{Z\in\C}2^{-\dim Z/2}>0,
\]
so a vanishing eigenvalue can only come from a nontrivial character; such an
eigenvalue produces a linear dependence among the normalized indicators.

\begin{lemma}
\label{lem:singular-dep}
If $M$ is singular, then there is a nontrivial character $\chi$ of
$(\C,\boxsum)$ with
\begin{equation}
 \sum_{X\in\C}\chi(X)\,2^{-\dim X/2}\,\mathbf 1_X=0
 \qquad\text{in }\mathbb R^{V}.
 \label{eq:dependence}
\end{equation}
\end{lemma}

\begin{proof}
By Proposition~\ref{prop:eigen} there is a character $\chi$ with $\mu_\chi=0$,
necessarily nontrivial by the preceding remark. Consider the vector
\[
 w=\sum_{X\in\C}\chi(X)\,u_X,
\]
with $u_X$ as in \eqref{eq:normalized}. Then, by \eqref{eq:gram-entry} and
$\chi(X)^2=1$,
\[
 \|w\|^2
 =\sum_{X,Y\in\C}\chi(X)\chi(Y)\langle u_X,u_Y\rangle
 =\mu_\chi\sum_{X\in\C}\chi(X)^2
 =0,
\]
so $w=0$. Since $u_X=2^{-\dim X/2}\mathbf 1_X$, this is exactly
\eqref{eq:dependence}.
\end{proof}

By Proposition~\ref{prop:eigen} and Lemma~\ref{lem:singular-dep}, the BEV
conjecture reduces to a single statement: no nontrivial character satisfies
\eqref{eq:dependence}, equivalently $\mu_\chi\ne0$ for every character $\chi$.
We prove this in the remainder of the section.

The coefficients in \eqref{eq:dependence} lie in the real quadratic field
$\mathbb Q(\sqrt2)$: we have $\chi(X)\in\{\pm1\}$, and, using
$2^{-1/2}=\tfrac12\sqrt2$,
\begin{equation}
 2^{-\dim X/2}=
 \begin{cases}
  2^{-\dim X/2}\in\mathbb Q, & \dim X\text{ even},\\[6pt]
  2^{-(\dim X+1)/2}\,\sqrt2\in\mathbb Q\,\sqrt2, & \dim X\text{ odd}.
 \end{cases}
 \label{eq:coeff-split}
\end{equation}
Since $\{1,\sqrt2\}$ is a $\mathbb Q$-basis of $\mathbb Q(\sqrt2)$, the rational
and $\sqrt2$ components of any $\mathbb Q(\sqrt2)$-linear relation vanish
separately; by \eqref{eq:coeff-split} the even-dimensional codewords contribute
the rational component and the odd-dimensional ones the $\sqrt2$ component.

These reductions give us all the ingredients needed to prove the BEV bound.

\begin{theorem}
\label{thm:main}
Every binary linear subspace code $\C\subseteq\PP_2(n)$ satisfies
$|\C|\le 2^n$; that is, Conjecture~\ref{conj:bev} holds.
\end{theorem}

\begin{proof}
By Proposition~\ref{prop:gram} it suffices to show that the Gram matrix $M$ is
nonsingular. We proceed by contradiction. Suppose $M$ is singular. By
Lemma~\ref{lem:singular-dep} there is a nontrivial character $\chi$ of
$(\C,\boxsum)$ with
\[
 \sum_{X\in\C}\chi(X)\,2^{-\dim X/2}\,\mathbf 1_X=0
 \qquad\text{in }\mathbb R^{V}.
\]
Taking the $v$-coordinate of this identity, and recalling that
$\mathbf 1_X(v)=1$ exactly when $v\in X$, we obtain, for every $v\in V$,
\begin{equation}
 \sum_{\substack{X\in\C\\ v\in X}}\chi(X)\,2^{-\dim X/2}=0 .
 \label{eq:coord}
\end{equation}
This is an equality in $\mathbb Q(\sqrt2)$. By \eqref{eq:coeff-split} its
rational part is the sum of the even-dimensional terms, so, taking rational
parts on both sides in \eqref{eq:coord},
\begin{equation}
 \sum_{\substack{X\in\C,\ v\in X\\ \dim X\ \mathrm{even}}}
 \chi(X)\,2^{-\dim X/2}=0
 \qquad(v\in V).
 \label{eq:rational}
\end{equation}

We exploit \eqref{eq:rational} for two summations over $v$. First we substitute $v=0$, and use the fact that every codeword contains $0$. The following expression is obtained by separating the zero codeword, whose
contribution is $\chi(0)\,2^{0}=1$,
\begin{equation}
 1+\sum_{\substack{X\in\C,\ X\ne0\\ \dim X\ \mathrm{even}}}
 \chi(X)\,2^{-\dim X/2}=0 .
 \label{eq:atzero}
\end{equation}

Summing \eqref{eq:rational} over all nonzero $v\in V$, a codeword of dimension
$k$ is counted once for each of its $2^{k}-1$ nonzero vectors, while the zero
codeword contributes nothing:
\begin{equation}
 \sum_{\substack{X\in\C,\ X\ne0\\ \dim X\ \mathrm{even}}}
 \bigl(2^{\dim X}-1\bigr)\,\chi(X)\,2^{-\dim X/2}=0 .
 \label{eq:sumv}
\end{equation}

Adding \eqref{eq:atzero} and \eqref{eq:sumv} and using
$1+\bigl(2^{\dim X}-1\bigr)=2^{\dim X}$ together with
$2^{\dim X}\cdot 2^{-\dim X/2}=2^{\dim X/2}$ gives
\begin{equation}
 1+\sum_{\substack{X\in\C,\ X\ne0\\ \dim X\ \mathrm{even}}}
 \chi(X)\,2^{\dim X/2}=0 .
 \label{eq:parity}
\end{equation}
Every codeword $X$ in the sum \eqref{eq:parity} has even dimension
$\dim X\ge 2$, so $2^{\dim X/2}$ is an even integer; as $\chi(X)=\pm1$, each
summand is an even integer, and hence so is the sum. The left-hand side of
\eqref{eq:parity} is therefore an odd integer, which cannot vanish. This
contradiction shows that $M$ is nonsingular, and so $|\C|\le 2^n$ by
Proposition~\ref{prop:gram}.
\end{proof}

The proof identifies the equality case. Since $M$ is nonsingular, it is
positive definite, so the normalized indicators
$\{\,2^{-\dim X/2}\,\mathbf 1_X : X\in\C\,\}$ are not merely linearly
independent but form an independent set of size $|\C|$ in the $2^n$-dimensional
space $\mathbb R^{V}$. Hence
\[
 \begin{aligned}
  |\C|=2^n
   &\iff \{\,2^{-\dim X/2}\,\mathbf 1_X : X\in\C\,\}\ \text{is a basis of }
         \mathbb R^{V}\\[4pt]
   &\iff \{\,\mathbf 1_X : X\in\C\,\}\ \text{is a basis of }\mathbb R^{V}.
 \end{aligned}
\]

The bound is attained. For a basis $\{e_1,\dots,e_n\}$ of $\F_2^n$, the spans
of all its subsets, with the symmetric difference of index sets as linear
addition, form a linear subspace code of size $2^n$~\cite[Theorem~6]{Braun};
this \emph{code derived from a fixed basis} contains the $n$ coordinate lines
$\langle e_i\rangle$ and the full space $\F_2^n$. It is not, however, the only
code meeting the bound. The following example, due to Braun, Etzion and
Vardy~\cite[Example~1]{Braun}, attains $2^n$ while containing neither a
one-dimensional codeword nor the full space, so equality does not force a code
to be derived from a fixed basis.

\begin{example}
\label{ex:equality}
Let $V=\F_2^3$. For a nonzero linear functional $f\in V^{*}$ write
$H_f=\ker f$, a two-dimensional subspace. Over $\F_2$ two nonzero functionals
have the same kernel only if they are equal, so the seven nonzero functionals
give the seven two-dimensional subspaces of $V$. Set
\[
 \C_0=\{0\}\cup\{\,H_f : f\in V^{*}\setminus\{0\}\,\},
\]
and define a linear addition by $0\boxsum H_f=H_f$ and
\[
 H_f\boxsum H_g=H_{f+g}\qquad(f,g\in V^{*}\setminus\{0\},\ f\ne g).
\]
The map $V^{*}\to\C_0$ given by $0\mapsto 0$ and $f\mapsto H_f$ is a bijection
carrying addition of functionals to $\boxsum$, so $(\C_0,\boxsum)$ is an
elementary abelian $2$-group of order $8$, with identity $0$ and every element
of order two.

For translation invariance it suffices, since $\boxsum$-translation permutes
$\C_0$, to check that $\C_0$ is equidistant. Every nonzero codeword has
dimension $2$, so $\dist(0,H_f)=2$. For distinct nonzero $f,g$, the linear map
$v\mapsto\bigl(f(v),g(v)\bigr)$ from $V$ to $\F_2^2$ is onto (as $f,g$ are then
linearly independent), so its kernel $H_f\cap H_g$ has dimension $1$, whence
\[
 \dist(H_f,H_g)=\dim H_f+\dim H_g-2\dim(H_f\cap H_g)=2+2-2=2.
\]
Thus all pairwise distances in $\C_0$ equal $2$; a constant distance is
automatically translation invariant, because translation preserves equality
and inequality of codewords. Hence $\C_0$ is a linear subspace code with
\[
 |\C_0|=8=2^3,
\]
meeting the bound of Theorem~\ref{thm:main}. Every nonzero codeword of $\C_0$
has dimension $2$, so $\C_0$ contains no one-dimensional codeword and does not
contain $V$; in particular it is not derived from a fixed basis. Basu, who
studies such equidistant codes systematically, identifies the nonzero
codewords of $\C_0$ with the lines of the Fano plane and shows that $\C_0$ is
the \emph{unique} equidistant linear code of size $2^3$ in
$\PP_2(3)$~\cite{Basu}.
\end{example}

Under the additional hypothesis that the full space is a codeword, Pai and Rajan~\cite{PaiRajan} proved that a code meeting the bound must be
derived from a fixed basis. Example~\ref{ex:equality} shows that this fails
without the hypothesis, and a classification of the equality codes in general
appears to be open.

%=====================================================================
\section{Concluding remarks}
\label{sec:remarks}
%=====================================================================
We have proved the Braun--Etzion--Vardy conjecture over $\F_2$ in full
generality (Theorem~\ref{thm:main}): every binary linear subspace code
satisfies $|\C|\le 2^n$, with no hypothesis on the code. The proof is a single
positive-definiteness computation for the indicator Gram matrix, and it
subsumes the earlier partial results of Pai and Rajan~\cite{PaiRajan},
who assumed $\F_2^n\in\C$, and of Mahak and Bhaintwal~\cite{MahakBhaintwal},
who counted one-dimensional codewords. Together with Basu's construction of
linear codes of size greater than $2^n$ for every prime power
$q>2$~\cite{Basu}, this shows that the bound $|\C|\le 2^n$ holds universally
precisely over the binary field.

One natural question remains open: the codes attaining the bound are not
classified. By Example~\ref{ex:equality} an equality code need not be derived
from a fixed basis, and the characterization of Pai and Rajan is known
only under the hypothesis $\F_2^n\in\C$.

\section*{Declaration of competing interest}
The author declares that he has no known competing financial interests or
personal relationships that could have appeared to influence the work reported
in this paper.

\section*{Funding}
This research did not receive any specific grant from funding agencies in the
public, commercial, or not-for-profit sectors.

\section*{Data availability}
No data was used for the research described in the article.

\section*{Declaration of generative AI and AI-assisted technologies in the manuscript preparation process}
During the development of this work, the author used ChatGPT Luna-5.6-medium as an interactive tool for mathematical exploration. The interaction involved proposing and testing proof strategies, with the author providing mathematical guidance, evaluating intermediate arguments, and rejecting incorrect or incomplete approaches. The final proof of Theorem~\ref{thm:main} emerged substantially from arguments proposed by the AI system during this process. The author subsequently checked the proof independently, revised its presentation, and takes full responsibility for its correctness and for all mathematical claims in the article.

%=====================================================================
%  Bibliography
%=====================================================================

\end{document}